\documentclass[conference,a4paper]{IEEEtran}
\IEEEoverridecommandlockouts
\usepackage[utf8]{inputenc} 
\usepackage[T1]{fontenc}
\usepackage{url}              
\usepackage{cite}             
\usepackage[cmex10]{amsmath}  
\usepackage{mleftright}       
\mleftright                   
\usepackage{booktabs}         
\usepackage{verbatim}
\usepackage{graphicx}
\usepackage{color, soul}
\usepackage{epstopdf}
\usepackage{subfigure} 
\usepackage{algorithmic}
\usepackage{amsthm}
\usepackage{extarrows}
\usepackage{mathrsfs}
\newtheorem{theorem}{Theorem}

\usepackage{bbm}
\usepackage[backref]{hyperref}
\usepackage{amssymb}

\newtheorem{property}{Property}
\newcommand{\RNum}[1]{\uppercase\expandafter{\romannumeral #1\relax}}
\newtheorem{lemma}{Lemma}

\usepackage{stfloats}

\usepackage{color}
\usepackage{tikz,xcolor,hyperref}

\definecolor{lime}{HTML}{A6CE39}
\allowdisplaybreaks[4]

\makeatletter
\DeclareRobustCommand{\orcidicon}{%
	\begin{tikzpicture}
		\draw[lime, fill=lime] (0,0) 
		circle [radius=0.16] 
		node[white] {{\fontfamily{qag}\selectfont \tiny ID}};    \draw[white, fill=white] (-0.0625,0.095) 
		circle [radius=0.007];    \end{tikzpicture}
	\hspace{-2mm}}
\foreach \x in {A, ..., Z}{%
	\expandafter\xdef\csname orcid\x\endcsname{\noexpand\href{https://orcid.org/\csname orcidauthor\x\endcsname}{\noexpand\orcidicon}}
}

\newcommand*\bigcdot{\mathpalette\bigcdot@{.5}}
\newcommand*\bigcdot@[2]{\mathbin{\vcenter{\hbox{\scalebox{#2}{$\m@th#1\bullet$}}}}}
\makeatother

\usepackage[linesnumbered,ruled,vlined]{algorithm2e}

\makeatletter
\renewcommand{\maketag@@@}[1]{\hbox{\m@th\normalsize\normalfont#1}}%
\makeatother
\begin{document}
\title{Tight Upper Bounds on the Error Probability of Spinal Codes over Fading Channels}

\author{%
  \IEEEauthorblockN{Xiaomeng~Chen, Aimin~Li and Shaohua~Wu}
  \IEEEauthorblockA{Harbin Institute of Technology (Shenzhen)\\
				    518055 Shenzhen, China\\
				    cxmeng@163.com, liaimin@stu.hit.edu.cn, hitwush@hit.edu.cn}

	\thanks{	
	Aimin Li has contributed equally to this work.  This work has been supported in part by the National Key Research and Development Program of China under Grant no. 2020YFB1806403, and in part by the National Natural Science Foundation of China under Grant nos. 61871147, 62071141, 61831008, 62027802,  and in part by the Shenzhen Municipal Science and Technology Plan under Grant no. GXWD20201230155427003-20200730122528002, and in part by the Major Key Project of PCL under Grant PCL2021A03-1. (Corresponding author: Shaohua Wu.)				
		}
	}				
		
		\maketitle
		\begin{abstract}
			Spinal codes, a family of rateless codes introduced in 2011, have been proved to achieve Shannon capacity over both the additive white Gaussian noise (AWGN) channel and the binary symmetric channel (BSC). In this paper, we derive explicit tight upper bounds on the error probability of Spinal codes under maximum-likelihood (ML) decoding and perfect channel state information (CSI) over three typical fading channels, including the Rayleigh channel, the Nakagami-m channel and the Rician channel. Simulation results verify the tightness of the derived upper bounds.
		\end{abstract}
		
		\IEEEpeerreviewmaketitle
		
		\section{Introduction} 
		First proposed in 2011 \cite{2011Spinal}, Spinal codes are a family of rateless codes that have been proved to achieve Shannon capacity over both the AWGN channel and the BSC \cite{balakrishnan2012randomizing}. Possessing the rateless capacity-achieving nature, Spinal codes demonstrate their superiority in bridging a reliable and high-efficiency information delivery pipeline between transceivers under highly dynamic channels. 
		Specifically, the rateless nature allows Spinal codes to automatically adapt to time-varying channel conditions. Unlike fixed-rate codes, which require a specific code rate in advance, rateless codes work by a natural channel adaptation manner: the sender transmits a potentially limitless stream of encoded bits, and the receiver collects bits consecutively until the successful decoding process takes place.
		In \cite{2012Spinal}, Spinal codes have shown advantages in terms of rate performance compared to other state-of-the-art rateless codes 
		under different channel conditions and message sizes. 
		Also, \cite{2012Spinal} notes the similarity between Spinal codes and Trellis Coded Modulation (TCM) \cite{TWC1,TWC2} is superficial because of their differences in nature and encoded popuse.
		
		In coding theory, performance analysis is an intriguing topic. A closed-form expression for the error probability, in general, could not only enable more efficient performance evaluations but also shed light on coding optimization design.
		However, in most cases, it is intractable to obtain a closed-form expression for error probabilities. 
		As an alternative, bounding techniques are usually used to approximate performance \cite{NBLLBC}. 
		Along this avenue, there are already a lot of established bounds for some specific channel codes, such as the advanced tight bounds on Polar codes \cite{polar1,polar2}, the upper and lower bounds on Raptor codes \cite{raptorerror}, and the performance bounds on the LT codes \cite{LTerror1,LTerror2}. 
		And in \cite{CIT-009}, upper and lower bounds on the error probability of linear codes under ML decoding are surveyed.
		For Spinal codes, however, it is still early days to get tight explicit bounds over fading channels. 
		
		To date, there have been a few works that have preliminarily explored the performance analysis of Spinal codes. In \cite{balakrishnan2012randomizing}, Balakrishnan \emph{et al.} analyze the asymptotic rate performance of Spinal codes and theoretically prove that Spinal codes are capacity-achieving over both the AWGN channel and the BSC. In \cite{UEPspinal}, two state-of-the-art results in the finite block length (FBL) regime, \emph{i.e.}, the Random Coding Union (RCU) bound \cite{polyanskiy2010channel} and a variant of Gallager result \cite{gallager1968information} are applied to analyze the error probability of Spinal codes over the AWGN channel and the BSC, respectively. In \cite{li2021new}, the authors further tighten bounds by characterizing the error probability as the volume of a hypersphere divided by the volume of a hypercube, while the analysis is still performed over the AWGN channel. Until now, little has been done in the way of error probability analysis for Spinel codes over fading channels. An exception work is \cite{li2021spinal}, which derived a probability-dependent convergent upper bound over Rayleigh fading channels by adopting the Chernoff bound \cite{chernoff1952measure}. However, $i$) the derived bound over the Rayleigh fading channel in \cite{li2021spinal} is not strictly explicit, since the convergence of the upper bound is probability-dependent; $ii$) Nakagami-m and Rician channels, both common in practical wireless communication scenarios, have not been considered in the available error probability analyses; $iii$) There is not yet an upper bound that achieves uniform tight approximations over a wide range of signal-to-noise ratio (SNR), either over the fading channel or over the AWGN channel.
		
		Motivated by the above, this paper aims to derive tight upper bounds over three typical fading channels, Rayleigh, Nakagami-m and Rician fading channels. Strictly explicit upper bounds are provided, which cover uniform tight error probability approximations over a wide range of SNR.
		
		The paper is organized as follows. In Section \ref{section II}, the encoding process of Spinal codes is given as a priori knowledge. Section \ref{section III} derives upper bounds over Rayleigh, Nakagami-m and
		\begin{figure}
			\centering
			\includegraphics[width=	0.9\linewidth]{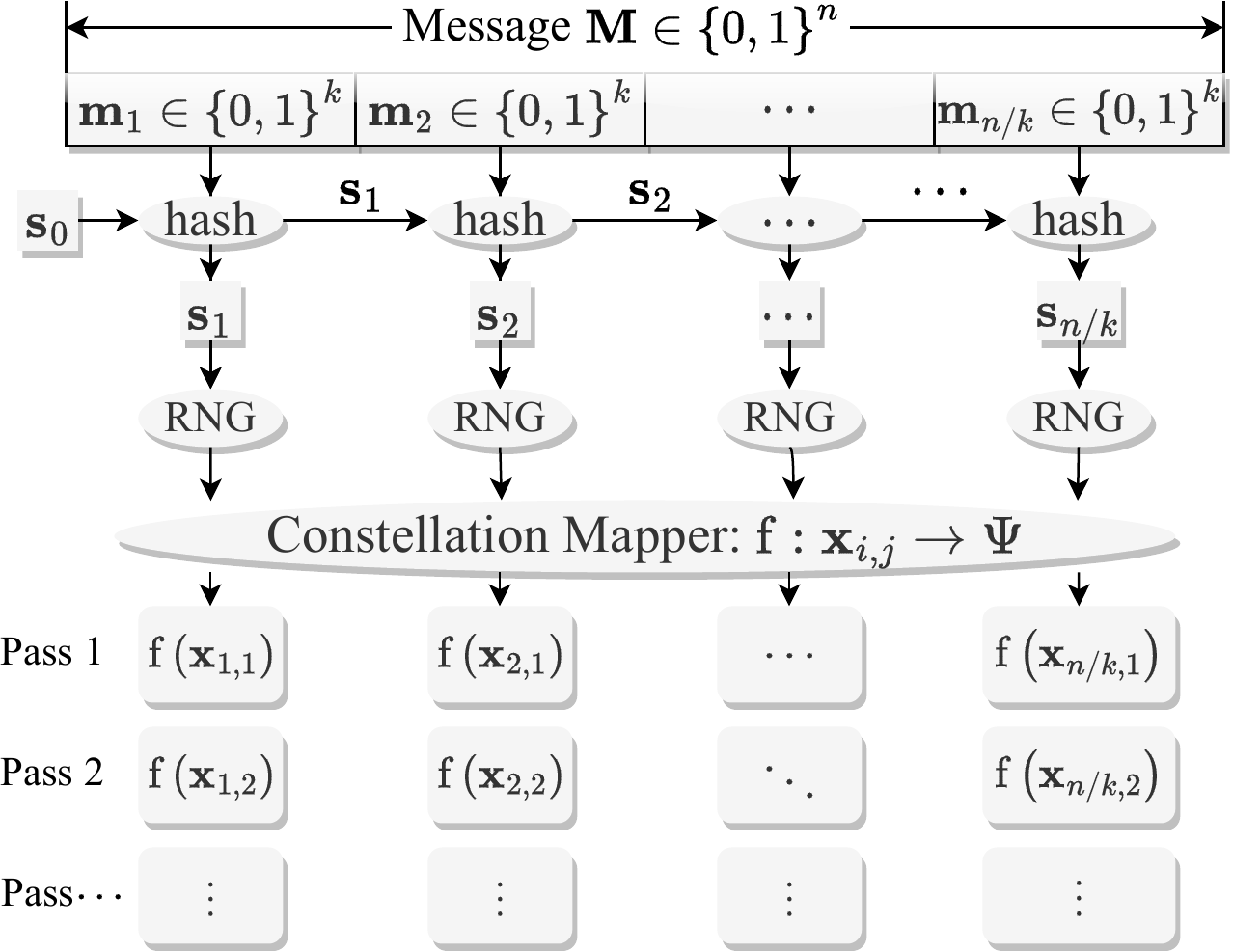}
			\caption{The encoding process of Spinal codes.}
			\label{figure1}
		\end{figure}
		Rician fading channels, respectively. Numerical results comparing the bounds with Monte Carlo simulations are illustrated in Section \ref{sectionIV}. Section \ref{sectionV} presents the conclusions of our work.
		
		\section{Encoding Process of Spinal Codes} \label{section II}
		This section briefly introduces the primary building blocks of Spinal codes, where the combination of a hash function and random number generator (RNG) functions is a key.  
		Fig. \ref{figure1} shows the encoding process of Spinal codes, which comprises five steps:
		\begin{enumerate}
			\item \textbf{Segmentation}: Divide an $n$-bit message $\mathbf{M}$ into $k$-bit segments $\mathbf{m}_i \in \left\{0,1 \right\}^k$, where $i=1,2,\cdots,n/k$.
			\item \textbf{Sequentially Hashing}\footnote{A discussion concerning the properties of the hash function can be found in Appendix \ref{property}.}: The hash function sequentially generates $v$-bit spine values $\mathbf{s}_i \in {\{0,1\}}^v$, with
			\begin{equation} \label{eqhash}
				\mathbf{s}_i = h(\mathbf{s}_{i-1},\mathbf{m}_i),~i=1,2,\cdots,n/k,~\mathbf{s}_0 = \mathbf{0}^v \text{.} 
				\footnote{The initial spine value $\mathbf{s}_0$ is known to both the encoder and the decoder, and is set as $\mathbf{s}_0=0$ in this papser without loss of generality.}
			\end{equation}
			\item \textbf{RNG}: Each spine value $\mathbf{s}_i$ is used to seed an RNG to generate a binary pseudo-random uniform-distributed sequence $\{\mathbf{x}_{i,j}\}_{j \in \mathbb{N}^+}$:
			\begin{equation}
				\mathrm{RNG:} ~\mathbf{s}_i \to \{ \mathbf{x}_{i,j} \},~j=1,2,3,\cdots \text{,}
			\end{equation}
			where $\mathbf{x}_{i,j} \in {\{ 0,1 \}}^c$.
			\item \textbf{Constellation Mapping}: The constellation mapper maps each $c$-bit symbol $\mathbf{x}_{i,j}$ to a channel input set $\Psi$:
			\begin{equation}
				\operatorname{f}: \mathbf{x}_{i,j}\rightarrow \Psi \text{.}
			\end{equation}
			In this paper, $\operatorname{f}$ is a uniform constellation mapping function, \emph{i.e.}, it converts each $c$-bit symbol $\mathbf{x}_{i,j}$ to a decimal symbol $\operatorname{f}(\mathbf{x}_{i,j})$.
			\item \textbf{Rateless Transmission}: The encoder continues to generate and transmit symbols pass by pass until the receiver successfully decodes the message.
		\end{enumerate}
		
		\section{UPPER BOUNDS ON THE ERROR PROBABILITY}\label{section III}				
		Error probability analysis is of primary importance in coding theory.
		Closed-form expressions for error probabilities, however, are intractable in most cases. To address this issue, a commonly adopted alternative approach is to introduce bounding techniques for the error probability approximation. In line with this idea, this section aims to derive tight upper bounds on the error probability for Spinal codes over Rayleigh, Nakagami-m and Rician fading channels, respectively.
		
		\subsection{Upper Bounds on the Rayleigh Fading Channel}
		\begin{theorem} \label{theorem1}
			Consider Spinal codes with message length $n$, segmentation parameter $k$, modulation parameter $c$, and sufficiently large hash parameter $v$ transmitted over a flat slow Rayleigh fading channel with mean square $\Omega$ and AWGN variance $\sigma^2$, the average error probability given perfect channel state information (CSI) under ML decoding for Spinal codes can be upper bounded by
			\begin{equation} \label{eq1}
				\small
				P_e \leq 1-\prod_{a=1}^{n/k} \left(1-\epsilon_a\right) \text{,}
			\end{equation}
			with
			\begin{equation} \label{eq2}
				\epsilon_a = \mathrm{min} \left\{ 1,\left(2^k-1\right)2^{n-ak} \cdot \mathscr{F}_{\text{Rayleigh}}(L_a,\sigma) \right\} \text{,}
			\end{equation}
			\begin{equation} \label{eq3}
				\small
				\mathscr{F}_{\text{Rayleigh}} \left(L_a \text{,} \sigma \right) = \sum_{r=1}^{N} b_r\mathcal{F}_{\text{Rayleigh}}(\theta_r \text{;} \sigma\text{,}L_a)\text{,}
			\end{equation}
			\begin{equation} \label{eq4}
				\small
				\mathcal{F}_{\text{Rayleigh}}(\theta_r\text{;} \sigma\text{,}L_a)= 
				{\left( \sum_{i \in \Psi} \sum_{j \in \Psi} 2^{-2c} \frac{8{\sigma}^2{\sin}^2{\theta}_r}{{\Omega}(i-j)^2+8{\sigma}^2{\sin}^2{\theta}_r} \right)}^{L_a}\text{,}
			\end{equation}
			where $b_r = \frac{\theta_r-\theta_{r-1}}{\pi}$, $L_a = (n/k-a+1)L$, $L$ is the number of transmitted passes, $\theta_r$ is arbitrarily chosen with $\theta_0=0$, $\theta_N=\frac{\pi}{2}$ and $0 < \theta_1 < \theta_2 < \cdots < \theta_{N-1} < \frac{\pi}{2}$, and $N$ represents the number of $\theta$ values which enables the adjustment of accuracy.
		\end{theorem}
		
		\begin{proof}
			Suppose a message $ \mathbf{M}=\left(\mathbf{m}_1,\mathbf{m}_2,\cdots,\mathbf{m}_{n/k}\right)\in {\{0,1\}}^n $ is encoded to $\operatorname{f}\left(\mathbf{x}_{i,j}(\mathbf{M})\right)$ to be transmitted over a flat slow Rayleigh fading channel with an AWGN. The received signal can be expressed as:
			\begin{equation} \label{eq5}
				y_{i,j} = h_{i,j}\operatorname{f}\left(\mathbf{x}_{i,j}(\mathbf{M})\right) + n_{i,j} \text{,}
			\end{equation}
			where $y_{i,j}$ is the corresponding received signal at the receiver, $h_{i,j}$ is the channel fading parameter following Rayleigh distribution with mean square $\Omega$,	and $n_{i,j}$ is the AWGN with variance $\sigma^2$.
			
			ML decoding	selects the one with the lowest decoding cost from the candidate sequence space ${\{0,1\}}^n$. Given received symbols $y_{i,j}$ and perfect CSI $\hat{h}_{i,j}=h_{i,j}$, the ML rule for the Rayleigh fading channel with the AWGN is	
			\begin{equation} \label{eq7}
				\small
				\begin{split}
					\hat{\mathbf{M}} &\in \underset{\mathbf{\bar{M}} \in {\{0,1\}}^n}{\mathrm{arg\,min}} \sum_{i=1}^{n/k} \sum_{j=1}^{L} {(y_{i,j}-h_{i,j}\operatorname{f}(\mathbf{x}_{i,j}(\bar{\mathbf{M}})))}^2 \text{,}
				\end{split}
			\end{equation}
			where $\hat{\mathbf{M}}$ is the decoding result, $\bar{\mathbf{M}}$ is the candidate sequence.
			
			From this perspective, we classify the set of candidate sequences ${\{0,1\}}^n$ into two  subsets: $i$) the correct decoding sequence $\mathbf{M} = \left(\mathbf{m}_1,\mathbf{m}_2,\cdots,\mathbf{m}_{n/k}\right)$; $ii$) wrong decoding sequences $\mathbf{M}' = \left(\mathbf{m}'_1,\mathbf{m}'_2,\cdots,\mathbf{m}'_{n/k}\right) \in \mathcal{W}$, with $\mathcal{W} \triangleq \left\{ \left(\mathbf{m}'_1,\mathbf{m}'_2,\cdots,\mathbf{m}'_{n/k}\right) : \exists 1 \leq i \leq n/k, \mathbf{m}'_i \neq \mathbf{m}_i \right\}$. Denoting the cost of $\mathbf{M}$ as $\mathscr{D}(\mathbf{M})$, it turns out that
			\begin{align} \label{eq8}
				\small
				\begin{split}
					\mathscr{D}(\mathbf{M}) &\triangleq \sum_{i=1}^{n/k}\sum_{j=1}^{L}{\left(y_{i,j}-h_{i,j}\operatorname{f}(\mathbf{x}_{i,j}(\mathbf{M}))\right)}^2 = \sum_{i=1}^{n/k}\sum_{j=1}^{L} n_{i,j}^2 \text{.}
				\end{split}
			\end{align}
			
			Similarly, denote the cost of $\mathbf{M}'$ as $\mathscr{D}(\mathbf{M}')$, given by:
			{
			\small
			\begin{align} \label{eq9}
				\mathscr{D}(\mathbf{M}') \triangleq \sum_{i=1}^{n/k}\sum_{j=1}^{L}{\left(y_{i,j}-h_{i,j}\operatorname{f}(\mathbf{x}_{i,j}(\mathbf{M}'))\right)}^2 \text{.}
			\end{align}
			}
			In the sequal, we attempt to explicitly express the error probability of Spinal codes. Let $\mathcal{E}_a$ be the event that there exists an error in the $a^{th}$ segment, which implies that:
			\begin{enumerate}
				\item The $a^{th}$ segment is different, \emph{i.e.}, $\mathbf{m}_a \neq \mathbf{m}'_a$.
				\item The cost of the wrong decoding sequence is less than the correct one, \emph{i.e.}, $\mathscr{D}(\mathbf{M}') \leq \mathscr{D}(\mathbf{M})$. In this case, the ML decoder will incorrectly choose a certain wrong sequence $\mathbf{M}'\in \mathcal{W}$ as the decoding output.
			\end{enumerate}
			
			Denote $\overline{\mathcal{E}}_a$ as the complement of $\mathcal{E}_a$. The error probability of Spinal codes can be expressed as:
			\begin{equation} \label{eq10}
				\small
				\begin{split}
					P_e &= \mathrm{Pr}\left( \bigcup_{a=1}^{n/k}\mathcal{E}_a \right) = 1 - \mathrm{Pr}\left(\bigcap_{a=1}^{n/k} \overline{\mathcal{E}}_a \right) \\
					&= 1 - \prod_{a=1}^{n/k} \left[ 1 - \mathrm{Pr} \left( \mathcal{E}_a\bigg|\bigcap_{i=1}^{a-1}\overline{\mathcal{E}}_i \right)\right] \text{.}
				\end{split}
			\end{equation}
			
			Thus, to obtain the error probability $P_e$, the remaining issue is to calculate $\mathrm{Pr} \bigg( \mathcal{E}_a\bigg|\bigcap_{i=1}^{a-1}\overline{\mathcal{E}}_i \bigg)$, which is interpreted as the probability that the $a^{th}$ segment is wrong while the previous $(a-1)$ segments are correct. With the defination that $\mathcal{W}_a \triangleq \{\left(\mathbf{m}'_1, \cdots ,\mathbf{m}'_a\right)\text{:} \mathbf{m}'_1=\mathbf{m}_1,\cdots,\mathbf{m}'_{a-1}=\mathbf{m}_{a-1},\mathbf{m}'_a \neq \mathbf{m}_a\} \subseteq \mathcal{W}$, the conditional probability is transformed as:
			\begin{equation} \label{eq11}
				\small
				\mathrm{Pr} \left( \mathcal{E}_a\bigg|\bigcap_{i=1}^{a-1}\overline{\mathcal{E}}_i \right) = \mathrm{Pr} \left( \exists \mathbf{M}'\in \mathcal{W}_a : \mathscr{D}(\mathbf{M}') \leq \mathscr{D}(\mathbf{M}) \right) \text{.}
			\end{equation}
			
			Applying the union bound of probability \cite{boole1847mathematical} yields that
			\begin{equation}
				\begin{split} \label{eq12}
					&\mathrm{Pr} \left( \exists \mathbf{M}'\in \mathcal{W}_a : \mathscr{D}(\mathbf{M}') \leq \mathscr{D}(\mathbf{M}) \right) \\
					&\leq \sum_{\mathbf{M}' \in \mathcal{W}_a} \mathrm{Pr} \left( \mathscr{D}(\mathbf{M}') \leq \mathscr{D}(\mathbf{M}) \right) \text{.}
				\end{split}
			\end{equation}
			
			Given that $\mathscr{D}(\mathbf{M})$ and $\mathscr{D}(\mathbf{M}')$ have been calculated in (\ref{eq8}) and (\ref{eq9}) respectively, the probability $\mathrm{Pr} \left( \mathscr{D}(\mathbf{M}') \leq \mathscr{D}(\mathbf{M}) \right)$ is calculated as follows:
			{
				\small
				\begin{align} \label{eq13}	
					& \mathrm{Pr} \left( \mathscr{D}(\mathbf{M}') \leq \mathscr{D}(\mathbf{M}) \right) \notag \\
					& = \mathrm{Pr}\bigg( \sum_{i=1}^{n/k}\sum_{j=1}^{L}{\left(y_{i,j}-h_{i,j}\operatorname{f}(x_{i,j}(\mathbf{M'}))\right)}^2 \leq \sum_{i=1}^{n/k}\sum_{j=1}^{L} n_{i,j}^2 \bigg) \notag\\
					& \overset{(a)}{=} \mathrm{Pr}\bigg( \sum_{i=a}^{n/k}\sum_{j=1}^{L}{\left(y_{i,j}-h_{i,j}\operatorname{f}(x_{i,j}(\mathbf{M'}))\right)}^2 \leq \sum_{i=a}^{n/k}\sum_{j=1}^{L} n_{i,j}^2 \bigg) \notag \\
					& \overset{(b)}{=} \mathrm{Pr}\bigg( \sum_{i=a}^{n/k}\sum_{j=1}^{L}{\left[h_{i,j}(\operatorname{f}(x_{i,j}(\mathbf{M}))-f(x_{i,j}(\mathbf{M'}))) + n_{i,j}\right]}^2 \notag \\
					& \qquad \qquad \qquad \qquad \leq \sum_{i=a}^{n/k}\sum_{j=1}^{L} n^2_{i,j} \bigg) ,
				\end{align}
			}
			where ($a$) establishes since $\operatorname{f}(\mathbf{x}_{i,j}(\mathbf{M}))=\operatorname{f}(\mathbf{x}_{i,j}(\mathbf{M'}))$ for $1\leq i<a$, which is proved in 
			Appendix \ref{property} by leveraging the property of hash function. ($b$) is obtained by applying (\ref{eq5}).
			
			Define $U_{i,j} \triangleq \operatorname{f}(\mathbf{x}_{i,j}(\mathbf{M}))-\operatorname{f}(\mathbf{x}_{i,j}(\mathbf{M'}))$ and $V_{i,j} \triangleq h_{i,j}U_{i,j}$, and then (\ref{eq13}) can be expanded as				
			\begin{equation} \label{eq15}
				\small
				\begin{split}
					&\mathrm{Pr}\bigg( \sum_{i=a}^{n/k}\sum_{j=1}^{L}{\bigg[\underbrace{h_{i,j}(\operatorname{f}(x_{i,j}(\mathbf{M}))-\operatorname{f}(x_{i,j}(\mathbf{M'})))}_{V_{i,j}} + n_{i,j}\bigg]}^2 \\
					& \qquad \qquad \quad \quad \leq \sum_{i=a}^{n/k}\sum_{j=1}^{L} n^2_{i,j} \bigg) \\
					&= \mathrm{Pr} \left( \sum_{i=a}^{n/k}\sum_{j=1}^{L} V^2_{i,j} + 2\sum_{i=a}^{n/k}\sum_{j=1}^{L} V_{i,j}n_{i,j} \leq 0 \right) \text{.}
				\end{split}
			\end{equation}
			
			Denoting $\mathbf{V}^{L_a}$ as the random row vector composed of random variables $V_{i,j}$ with $a \leq i \leq n/k, 1 \leq j \leq L$, and $\mathbf{N}^{L_a}$ as the random row vector composed of random variables $n_{i,j}$ with $a \leq i \leq n/k, 1 \leq j \leq l_i$, (\ref{eq15}) could be rewritten into a vector form, given as
			\begin{equation} \label{eq16}
				\small
				\begin{split}
					& \mathrm{Pr} \left( \sum_{i=a}^{n/k}\sum_{j=1}^{L} V^2_{i,j} + 2\sum_{i=a}^{n/k}\sum_{j=1}^{L} V_{i,j}n_{i,j} \leq 0 \right) \\
					&= \mathrm{Pr} \left( \mathbf{V}^{L_a} {\left(\mathbf{V}^{L_a} + 2{\mathbf{N}}^{L_a}\right)}^{\mathrm{T}} \leq 0 \right),
				\end{split}
			\end{equation}
			which could be then expanded as
			\begin{equation}\label{integral19}
				\begin{aligned}
					&\int_{\mathbb{R}^{L_a}}\mathrm{Pr} \left( \mathbf{V}^{L_a} {\left(\mathbf{V}^{L_a} + 2{\mathbf{N}}^{L_a}\right)}^{\mathrm{T}} \leq 0 \big| \mathbf{V}^{L_a}=\mathbf{v}^{L_a} \right)\cdot\\ 
					&\quad \quad \quad \quad \quad \quad \quad \quad \quad \quad \quad \quad \mathrm{Pr} \left( \mathbf{V}^{L_a}=\mathbf{v}^{L_a} \right) {~\mathrm{d}\mathbf{v}^{L_a}}\\
					&=  \int_{\mathbb{R}^{L_a}}\mathrm{Pr} \left( \mathbf{v}^{L_a} {\left(\mathbf{v}^{L_a} + 2{\mathbf{N}}^{L_a}\right)}^{\mathrm{T}} \leq 0  \right)\cdot\\ 
					&\quad \quad \quad \quad \quad \quad \quad \quad \quad \quad \quad \quad  \mathrm{Pr} \left( \mathbf{V}^{L_a}=\mathbf{v}^{L_a} \right) {~\mathrm{d}\mathbf{v}^{L_a}}.
				\end{aligned}
			\end{equation}			
			
			With (\ref{integral19}) in hand, the next problem is to explicitly solve $\mathrm{Pr} \left( \mathbf{v}^{L_a} {\left(\mathbf{v}^{L_a} + 2{\mathbf{N}}^{L_a}\right)}^{\mathrm{T}} \leq 0 \right)$ and $\mathrm{Pr} \left( \mathbf{V}^{L_a}=\mathbf{v}^{L_a} \right)$. First, we attempt to simplify $\mathrm{Pr} \left( \mathbf{v}^{L_a} {\left(\mathbf{v}^{L_a} + 2{\mathbf{N}}^{L_a}\right)}^{\mathrm{T}} \leq 0 \right)$. By introducing two rotation matrices for $L_a$-dimensions hyperspace, we obtain Lemma \ref{lemma1} as follows:
			\begin{lemma}\label{lemma1}
				Given that $n_{i,j}$ is i.i.d AWGN with variance $\sigma^2$, the probability in (\ref{integral19}) can be simplified as
				\begin{equation}\label{Qtransformation}
					\small
					\mathrm{Pr} \left( \mathbf{v}^{L_a} {\left(\mathbf{v}^{L_a} + 2{\mathbf{N}}^{L_a}\right)}^{\mathrm{T}} \leq 0 \right)=Q\left( \frac{\left\|\mathbf{v}^{L_a}\right\|}{2\sigma} \right),
				\end{equation}
				where $Q(\cdot)$ represents the $Q$ function, $\left\|\cdot\right\|$ represents the $\ell^2$ norm.
			\end{lemma}
			\begin{proof}
				Please refer to Appendix \ref{rotationmatrix}.
			\end{proof}
			Note that \cite{Qfunction} has shown a transformation of $Q(\cdot)$, which is presented as an exponential form, given as
			\begin{equation} \label{eq33}
				\small
				Q(x) = \frac{1}{\pi} \int_{0}^{\frac{\pi}{2}} \mathrm{exp} \left(\frac{-x^2}{2\mathrm{sin}^2\theta}\right) \mathrm{d}\theta \text{.}
			\end{equation}
			Adopting (\ref{eq33}) into (\ref{Qtransformation}) yields that
			\begin{align} \label{eq34}
				\small
					& \mathrm{Pr} \left( \mathbf{v}^{L_a} {\left(\mathbf{v}^{L_a} + 2{\mathbf{N}}^{L_a}\right)}^{\mathrm{T}} \leq 0 \right) \notag\\
					&= Q\left( \frac{\left\|\mathbf{v}^{L_a}\right\|}{2\sigma} \right) = \frac{1}{\pi} \int_{0}^{\frac{\pi}{2}} \mathrm{exp} \left(\frac{-{\left\|\mathbf{v}^{L_a}\right\|}^2}{8 \sigma^2 \mathrm{sin}^2\theta}\right) \mathrm{d}\theta \text{.}
			\end{align}
			Applying (\ref{eq34}) in (\ref{integral19}) and swapping the integrates, we have
			\begin{equation} \label{eq35}
				\small
				\begin{split}
					&\mathrm{Pr} \left( \mathbf{V}^{L_a} {\left(\mathbf{V}^{L_a} + 2{\mathbf{N}}^{L_a}\right)}^{\mathrm{T}} \leq 0 \right)\\
					&=\frac{1}{\pi}\int_{\mathbb{R}^{L_a}} \int_{0}^{\frac{\pi}{2}} \mathrm{exp} \left(\frac{-{\left\|\mathbf{v}^{L_a}\right\|}^2}{8 \sigma^2 \mathrm{sin}^2\theta}\right) \mathrm{d}\theta \cdot \mathrm{Pr} \left( \mathbf{V}^{L_a}=\mathbf{v}^{L_a} \right) \mathrm{d}\mathbf{v}^{L_a}\\
					&= \frac{1}{\pi}\int_{0}^{\frac{\pi}{2}}\underbrace{\int_{\mathbb{R}^{L_a}}  \mathrm{exp} \left(\frac{-{\left\|\mathbf{v}^{L_a}\right\|}^2}{8 \sigma^2 \mathrm{sin}^2\theta}\right)  \cdot \mathrm{Pr} \left( \mathbf{V}^{L_a}=\mathbf{v}^{L_a} \right) ~\mathrm{d}\mathbf{v}^{L_a}}_{\mathcal{F}_{\text{Rayleigh}(\theta;\sigma,L_a)}} \mathrm{d}\theta\text{.}
				\end{split}
			\end{equation}
		
			For the Rayleigh fading channel, we denote the multiple integrals with respect to $\mathbf{v}^{L_a}$ as $\mathcal{F}_{\text{Rayleigh}}\left(\theta;\sigma,L_a\right)$. By adopting the i.i.d of $V_{i,j}$, given as
			\begin{equation}
				\small
				\Pr\left(\mathbf{V}^{L_a}=\mathbf{v}^{L_a}\right)=\prod_{i=a}^{n/k}\prod_{j=1}^{L}f_{V_{i,j}}\left(v_{i,j}\right),
			\end{equation}
			$\mathcal{F}_{\text{Rayleigh}}\left(\theta;\sigma,L_a\right)$ could be decomposed as
			\begin{equation} \label{eq36}
				\small
				\begin{split}
					& \underbrace{\int_{\mathbb{R}} \cdots \int_{\mathbb{R}}}_{L_a} \mathrm{exp} \left(\frac{-{\left\|\mathbf{v}^{L_a}\right\|}^2}{8 \sigma^2 \mathrm{sin}^2\theta}\right) \prod_{i=a}^{n/k}\prod_{j=1}^{L} f_{V_{i,j}}(v_{i,j}) \prod_{i=a}^{n/k}\prod_{j=1}^{L} \mathrm{d}v_{i,j} \\
					&\overset{(a)}{=} \prod_{i=a}^{n/k}\prod_{j=1}^{L} \int_{\mathbb{R}} \mathrm{exp} \left(\frac{-v_{i,j}^2}{8 \sigma^2 \mathrm{sin}^2\theta}\right) f_{V_{i,j}}(v_{i,j})~\mathrm{d}v_{i,j} \\
					&\overset{(b)}{=} { \left( \int_{\mathbb{R}} {\mathrm{exp}}{\left(-\frac{v_{a,1}^2}{{8}{\sigma^2}{\sin^2}\theta}\right)} f_{V_{a,1}}(v_{a,1}) \mathrm{d}v_{a,1} \right) }^{L_a} \text{,}
				\end{split}
			\end{equation}
			where (a) establishes by adopting
			\begin{equation}
				\small
				\begin{split}
					\exp\left(\frac{-\|\mathbf{v}^{L_a}\|^2}{8\sigma^2 \sin^2 \theta}\right)&=\exp\left(\frac{-\sum_{i=a}^{n/k}\sum_{j=1}^{L}v^2_{i,j}}{8\sigma^2 \sin^2 \theta}\right)\\
					&=\prod_{i=a}^{n/k}\prod_{j=1}^{L}\exp\left(\frac{-v^2_{i,j}}{8\sigma^2 \sin^2 \theta}\right) \text{,}
				\end{split}
			\end{equation}
			and (b) holds for the i.i.d of the random variable $V_{i,j}$. Recall that $V_{i,j} = h_{i,j}U_{i,j}$ where $h_{i,j}$ and $U_{i,j}$ are independent with each other, the integral in (\ref{eq36}) with respect to $v_{a,1}$ could be then transformed to 
			\begin{equation} \label{eqreplace}
				\small
				\begin{split}
					& \int_{\mathbb{R}} {\mathrm{exp}}{\left(-\frac{v_{a,1}^2}{{8}{\sigma^2}{\sin^2}\theta}\right)} f_{{V}_{a,1}}(v_{a,1}) \mathrm{d}v_{a,1} \\
					& \xlongequal{v_{a,1} = hu} \sum_{u} p_{_U}(u) \int_{\mathbb{R}} {\mathrm{exp}}{\left(-\frac{{h^2}{u^2}}{{8}{\sigma^2}{\sin^2}\theta}\right)} g_{_1}(h)\mathrm{d}h \text{,}
				\end{split}
			\end{equation} 
			where $p_{_U}(u)$ is the probability mass function (PMF) of $U_{a,1}$ and $g_{_1}(h)$ is the probability density function (PDF) of $h$. At this point, the right-hand side (RHS) of (\ref{eqreplace}) could be explicitly calculated by the following lemma.
			
			\begin{lemma} \label{lemma2}
				For Rayleigh distribution whose PDF is $g_{_1}(h)$, the RHS of (\ref{eqreplace}) can be calculated by
				\begin{equation} \label{eqq37}
					\small
					\sum_{i \in \Psi} \sum_{j \in \Psi} 2^{-2c} \frac{8{\sigma}^2{\sin}^2{\theta}}{{\Omega}(i-j)^2+8{\sigma}^2{\sin}^2{\theta}} \text{,}
				\end{equation} 
				where $\Psi$ is the channel input set.
			\end{lemma}
			\begin{proof}
				Please refer to Appendix \ref{proofray} for the detailed derivation.
			\end{proof}

			As such, substituting (\ref{eqq37}) back into (\ref{eq36}) turns out that
			\begin{equation} \label{eq37}
				\small
				\mathcal{F}_{\text{Rayleigh}}(\theta\text{;} \sigma\text{,}L_a) =
				{\left( \sum_{i \in \Psi} \sum_{j \in \Psi} 2^{-2c} \frac{8{\sigma}^2{\sin}^2{\theta}}{{\Omega}(i-j)^2+8{\sigma}^2{\sin}^2{\theta}} \right)}^{L_a} \text{.}
			\end{equation}
			With (\ref{eq37}) in hand, the RHS of (\ref{eq12}) is transformed as
			\begin{equation} \label{eq38}
				\small
				\begin{split}
					& \sum_{\mathbf{M}' \in \mathcal{W}_a} \mathrm{Pr} \left( \mathscr{D}\left(\mathbf{M}'\right) \leq \mathscr{D}\left(\mathbf{M}\right) \right) \\
					&= \frac{1}{\pi} \sum_{\mathbf{M}' \in \mathcal{W}_a} \int_{0}^{\frac{\pi}{2}} \mathcal{F}_{\text{Rayleigh}}(\theta\text{;} \sigma\text{,}L_a) \mathrm{d}\theta \text{.}
				\end{split}
			\end{equation}
			
			Note that $\mathcal{F}_{\text{Rayleigh}}(\theta\text{;} \sigma\text{,}L_a)$ is inreasing with $\theta$ for $0 \leq \theta \leq \frac{\pi}{2}$ (see Appendix \ref{monotony} for the detailed proof), so we can arbitrarily choose $N+1$ values of $\theta$ such that $\theta_0=0$, $\theta_N=\frac{\pi}{2}$ and $0 < \theta_1 < \theta_2 < \cdots < \theta_{N-1} < \frac{\pi}{2}$ to explicitly upper bound the error probability. Combining these values back into (\ref{eq38}), we get the inequality as
			\begin{equation} \label{eq39}
				\small
				\begin{split}
					&\sum_{\mathbf{M'} \in \mathcal{W}_a} \mathrm{Pr} \left( \mathscr{D}\left(M'\right) \leq \mathscr{D}\left(M\right) \right) \\
					&= \frac{1}{\pi} \cdot \sum_{\mathbf{M}' \in \mathcal{W}_a} \int_{0}^{\frac{\pi}{2}} \mathcal{F}_{\text{Rayleigh}}(\theta\text{;} \sigma\text{,}L_a) \mathrm{d}\theta \\
					&\le \frac{\left| \mathcal{W}_a \right|}{\pi} \cdot \sum_{r=1}^{N} \int_{\theta_{r-1}}^{\theta_r} \mathcal{F}_{\text{Rayleigh}}(\theta_r\text{;} \sigma\text{,}L_a) \mathrm{d}\theta \\
					&= \left| \mathcal{W}_a \right| \cdot \sum_{r=1}^{N} b_r \mathcal{F}_{\text{Rayleigh}}(\theta_r\text{;} \sigma\text{,}L_a) \text{,}
				\end{split}
			\end{equation}
			where
			\begin{equation} \label{eq40}
				\small
				b_r = \frac{\theta_r-\theta_{r-1}}{\pi} \text{,} \quad\left| \mathcal{W}_a \right| = (2^k-1)2^{n-ak}.
			\end{equation}
		
			Denoting
			\begin{equation} \label{eq41}
				\small
				\mathscr{F}_{\text{Rayleigh}} \left(L_a \text{,} \sigma \right) = \sum_{r=1}^{N} b_r\mathcal{F}_{\text{Rayleigh}}(\theta_r \text{;} \sigma\text{,}L_a)\text{,}
			\end{equation}
			and applying (\ref{eq41}) to (\ref{eq12}) results in the explicit bound.
		\end{proof}
	
		\subsection{Bounds on the Nakagami-m Fading Channel}
		\begin{theorem} \label{theorem2}
			Consider Spinal codes with message length $n$, segmentation parameter $k$, modulation parameter $c$ and sufficiently large hash parameter $v$ transmitted over a flat slow Nakagami-m fading channel with mean square $\Omega$, shape parameter $m$ and AWGN variance $\sigma^2$. The average error probability given perfect CSI under ML decoding for Spinal codes can be upper bounded by
			\begin{align} \label{eq42}
				\small
				P_e \leq 1-\prod_{a=1}^{n/k} \left(1-\epsilon_a\right) \text{,}
			\end{align}
			with	
			\begin{equation} \label{eq43}
				\epsilon_a = \mathrm{min} \left\{ 1,\left(2^k-1\right)2^{n-ak} \cdot \mathscr{F}_{\text{Nakagami}}(L_a,\sigma) \right\} \text{,}
			\end{equation}
			\begin{equation} \label{eq44}
				\mathscr{F}_{\text{Nakagami}} \left(L_a \text{,} \sigma \right) = \sum_{r=1}^{N} b_r\mathcal{F}_{\text{Nakagami}}(\theta_r \text{;} \sigma\text{,}L_a)\text{,}
			\end{equation}
			\begin{equation} \label{eq45}
				\small
				\begin{split}
					& \mathcal{F}_{\text{Nakagami}}(\theta_r\text{;} \sigma\text{,}L_a) \\
					&= {\left( \sum_{i \in \Psi} \sum_{j \in \Psi} 2^{-2c}
						{\left(\frac{8m{\sigma}^2{\sin}^2{\theta}_r}{{\Omega}(i-j)^2+8m{\sigma}^2{\sin}^2{\theta}_r}\right)}^m \right)}^{L_a}\text{,}
				\end{split}
			\end{equation}
			where $b_r = \frac{\theta_r-\theta_{r-1}}{\pi}$, $L_a = (n/k-a+1)L$, $L$ is the number of transmitted passes, $\theta_r$ is arbitrarily chosen with $\theta_0=0$, $\theta_N=\frac{\pi}{2}$ and $0 < \theta_1 < \theta_2 < \cdots < \theta_{N-1} < \frac{\pi}{2}$, and N represents the number of $\theta$ values which enables the adjustment of accuracy.
		\end{theorem}
		
		\begin{proof}
			Please refer to Appendix \ref{proofnaka}.
		\end{proof}
		
		
		\subsection{Bounds on the Rician Fading Channel}
		\begin{theorem} \label{theorem3}
			Consider Spinal codes with message length $n$, segmentation parameter $k$, modulation parameter $c$ and sufficiently large hash parameter $v$ transmitted over a flat slow Rician fading channel with mean square $\Omega$ ,shape parameter $K$ and AWGN variance $\sigma^2$. The average error probability given perfect CSI under ML decoding for Spinal codes can be upper bounded by
			\begin{equation} \label{eq52}
				\small
				P_e \leq 1-\prod_{a=1}^{n/k} \left(1-\epsilon_a\right) \text{,}
			\end{equation}
			with
			\begin{equation} \label{eq53}
				\epsilon_a = \mathrm{min} \left\{ 1,\left(2^k-1\right)2^{n-ak} \cdot \mathscr{F}_{\text{Rician}}(L_a,\sigma) \right\} \text{,}
			\end{equation}
			\begin{equation} \label{eq54}
				\mathscr{F}_{\text{Rician}} \left(L_a \text{,} \sigma \right) = \sum_{r=1}^{N} b_r\mathcal{F}_{\text{Rician}}(\theta_r \text{;} \sigma\text{,}L_a)\text{,}
			\end{equation}
			\begin{equation} \label{eq55}
				\small
				\begin{split}
					& \mathcal{F}_{\text{Rician}}(\theta_r\text{;} \sigma\text{,}L_a) \\
					&= \Bigg( \sum_{i \in \Psi} \sum_{j \in \Psi} 2^{-2c} \frac{8(K+1){\sigma}^2{\sin}^2{\theta}_r}{{\Omega}(i-j)^2+8(K+1){\sigma}^2{\sin}^2{\theta}_r} \\
					& \qquad \cdot {\mathrm{exp} \left( \frac{8K(K+1){\sigma}^2{\sin}^2{\theta}_r}{{\Omega}(i-j)^2+8(K+1){\sigma}^2{\sin}^2{\theta}_r} -K\right) \Bigg)}^{L_a}\text{,}
				\end{split}
			\end{equation}
			where $b_r = \frac{\theta_r-\theta_{r-1}}{\pi}$, $L_a = (n/k-a+1)L$, $L$ is the number of transmitted passes, $\theta_r$ is arbitrarily chosen with $\theta_0=0$, $\theta_N=\frac{\pi}{2}$ and $0 < \theta_1 < \theta_2 < \cdots < \theta_{N-1} < \frac{\pi}{2}$, and N represents the number of $\theta$ values which enables the adjustment of accuracy.
		\end{theorem}
		
		\begin{proof}
			Please refer to Appendix \ref{proofrice}.
		\end{proof}
		
		\section{Simulation Result}\label{sectionIV}
		In this section, we conduct Monte Carlo simulations to illustrate the accuracy of the upper bounds derived in Section \ref{section III}. Since the exponential nature of the ML-decoding complexity, the input message bits $n$ is selected as low as $n=8$ and the number of pass is set as $L = 6$ here for the easy ML-decoding simulation setup. The parameter $v$ is set to $v = 32$ for the implementation and experiments, demonstrated by the property \ref{property2} in Appendix \ref{property} that the hash collision will occur only once per $2^{32}$ hash function invocations on average. Furthermore, to improve the accuracy of approximation for upper bounds, we set $N=20$ and the sample size of Monte Carlo simulations to be $10^6$. 
		
		We nomalize the mean square values of all fading channels, \emph{i.e.}, $\Omega = 1$. Besides the setting of $\Omega$, for the Nakagami-m fading channel, we also set the shape parameter as $m=2$. And for Rician fading channels, the shape parameter $K$, defined as the ratio of the power contributions by line-of-sight path to the remaining multipaths, is set as $K=0.5$ and $K=1$, respectively.
		
		\begin{figure}
			\centering
			\includegraphics[width=	0.95\linewidth]{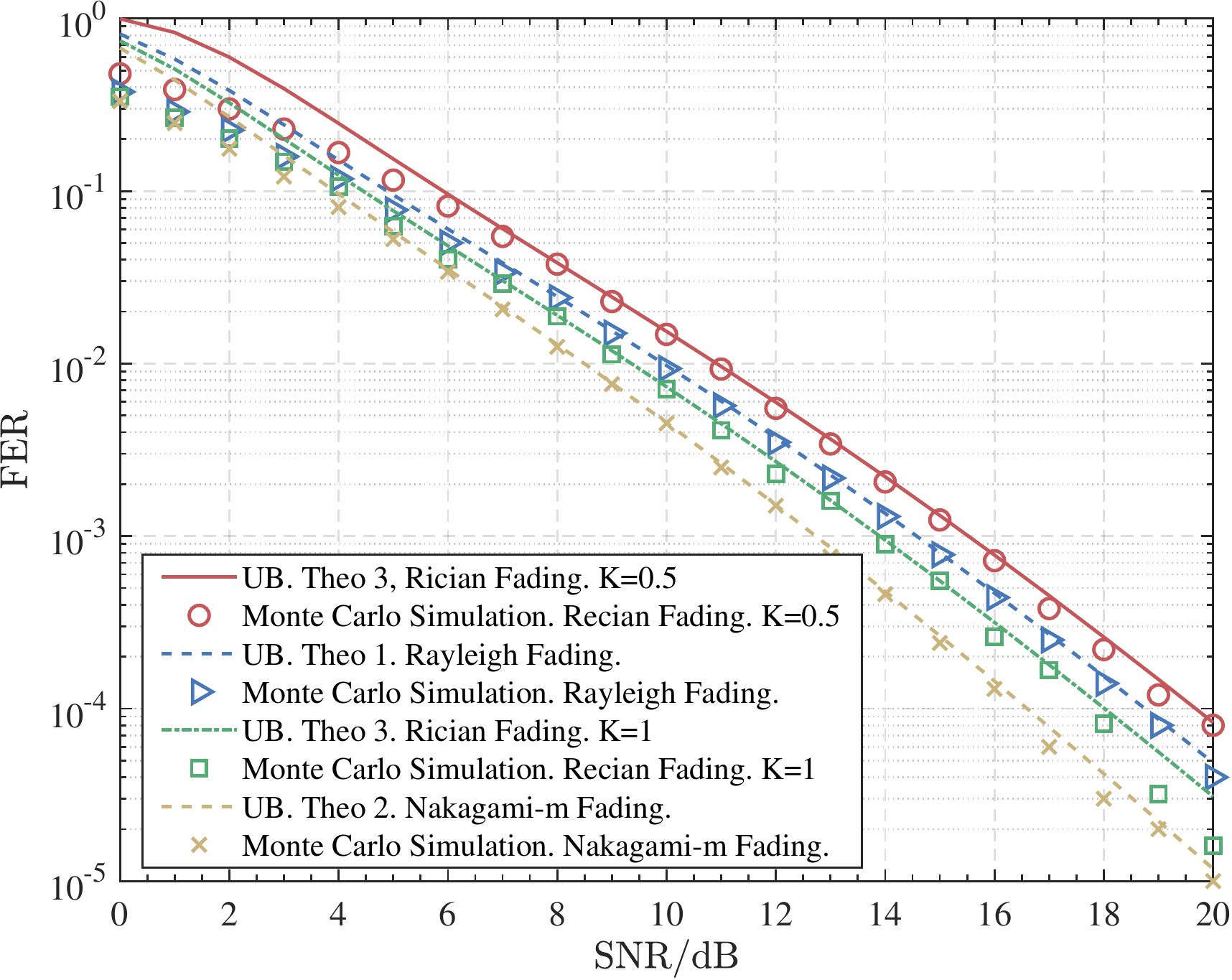}
			\caption{ Upper bounds on the error probability of Spinal codes with $n = 8, v = 32, L = 6, c = 8$ and $k = 2$ over several fading channels with $\Omega = 1$.}
			\label{figure2}
		\end{figure}
		
		Fig.\ref{figure2} shows examples of the error probability for Rayleigh, Nakagami-m and Rician fading channels, respectively. All approximations are close to simulated values. We could observe that the FER curve for $K=1$ is lower than the one for $K=0.5$ in Rician fading channels, which is due to $K$ is the ratio of the energy in the specular path to the energy in the scattered paths, \emph{i.e.}, the larger $K$, the more deterministic the channel.
		
		\section{Conclusion}\label{sectionV}
		This paper analyzes the error probability of Spinal codes and derives upper bounds on the error probability under ML decoding over several fading channels, including the Rayleigh fading channel, the Nakagami-m fading channel and the Rician fading channel. Additionally, we conduct simulations for different fading channels and parameters. Our experimental examples show that the upper bounds we derived have a good performence on the estimation of the average error probablity.
		
		The work in this paper may also inspire innovation about further research and efforts in related topics. The derived upper bounds can provide theoretical support and guidance for designing other high-efficiency coding-associated techniques, such as unequal error protection and concatenation with outer codes.
		
		\bibliographystyle{IEEEtran}
		\newpage
		\bibliography{reference}

		\clearpage
		\begin{appendices}
			\section{The properties of hash and related inferences}\label{property}
			The core idea of Spinal codes is to apply a hash function sequentially to generate a pseudo-random mapping from message bits to coded bits.  The hash function is expressed as
			\begin{equation}
				h: \{0,1\}^v \times \{0,1\}^k \rightarrow \{0,1\}^v,
			\end{equation}
			where $v$ and $k$ are both integers.
			In Spinal codes, $h$ is chosen from a pairwise independent family of hash functions $\mathscr{H}$\cite{hashfunction} by uniformly using a random seed.
			\begin{property} \label{property1}
				As Perry et al. indicate \cite{2012Spinal}, the hash function employed by Spinal codes should have pairwise independent property:
				\begin{equation}
					\begin{split}
						& \mathrm{Pr} \{ h(\mathbf{s},\mathbf{m}) = \mathbf{x},h(\mathbf{s}',\mathbf{m}') = \mathbf{x}' \} \\
						&= \mathrm{Pr} \{ h(\mathbf{s},\mathbf{m}) = \mathbf{x}\} \cdot \mathrm{Pr}\{ h(\mathbf{s}',\mathbf{m}') = \mathbf{x}' \} \\
						&= 2^{-2v} \text{,}
					\end{split}
				\end{equation}
				where $(\mathbf{s},\mathbf{m})$ and $(\mathbf{s}',\mathbf{m}')$ are arbitrarily chosen, $(\mathbf{s},\mathbf{m}) \neq (\mathbf{s}',\mathbf{m}')$.
			\end{property}
			
			\begin{property} \label{property2}
				A sufficiently large v will lead to a zero-approaching hash collision probability, with \cite{2012Spinal}
				\begin{equation}
					\begin{split}
						& \mathrm{Pr} \{ h(\mathbf{s},\mathbf{m}) = h(\mathbf{s}',\mathbf{m}') \} \\
						&= \sum_{\mathbf{x} \in {\{ 0,1 \}}^v} \underbrace{\mathrm{Pr} \{ h(\mathbf{s},\mathbf{m}) = \mathbf{x},h(\mathbf{s}',\mathbf{m}') = \mathbf{x}' \}}_{\mathrm{Property} \ref{property1}} \\
						&= 2^v \cdot 2^{-2v} \\
						&= 2^{-v} \approx 0,\ \text{iff}\ v \gg 0 \text{,}
					\end{split}
				\end{equation}
				where $(\mathbf{s},\mathbf{m})$ and $(\mathbf{s}',\mathbf{m}')$ are arbitrarily chosen, $(\mathbf{s},\mathbf{m}) \neq (\mathbf{s}',\mathbf{m}')$.
			\end{property}
			Recall that $\mathbf{M}$ is the correct decoding sequence and $\mathbf{M}' \in \mathcal{W}$ is a certain wrong decoding sequence. We denote the spine values of $\mathbf{M}$ as $s(\mathbf{M}) = (\mathbf{s}_1, \mathbf{s}_2, \cdots, \mathbf{s}_{n/k})$ and the spine values of $\mathbf{M}'$ as $s(\mathbf{M}') = (\mathbf{s}'_1, \mathbf{s}'_2, \cdots, \mathbf{s}'_{n/k})$.
			
			In (\ref{eqhash}), the process of generating $\mathbf{s}_i$ has given that $\mathbf{s}_i = h(\mathbf{s}_{i-1},\mathbf{m}_i)$, $\mathbf{s}_0=\mathbf{0}^v$. So it is natural to get that
			\begin{equation}
				\mathrm{Pr}( \mathbf{s}_i = \mathbf{s}'_i | \mathbf{m}_i = \mathbf{m}'_i, \mathbf{s}_{i-1} = \mathbf{s}'_{i-1} ) = 1\text{,}
			\end{equation}
			for $1 \leq i \leq n/k $.
			
			Conversely, if $\mathbf{m}_i \neq \mathbf{m}'_i$ or $\mathbf{s}_{i-1} \neq \mathbf{s}'_{i-1} $ with a sufficiently large $v$, from Property \ref{property2} we can obtain that $\mathbf{s}_i \neq \mathbf{s}'_i$ for $1 \leq i \leq n/k$.

			When $\mathbf{M}' \in \mathcal{W}_a$, $i.e.$, $\mathbf{m}_i = \mathbf{m}'_i$ for $1 \leq i < a$ and $\mathbf{m}_a \neq \mathbf{m}'_a$,  we can iterate to obtain that $\mathbf{s}_i = \mathbf{s}'_i$ for $1 \leq i < a$ and $\mathbf{s}_i \neq \mathbf{s}'_i$ for $a \leq i \leq n/k$. Hence after RNGs and uniform constellation mapping, the coded symbols satisfy
			\begin{equation} \label{eq(a)}
				\operatorname{f}(\mathbf{x}_{i,j}(\mathbf{M})) = \operatorname{f}(\mathbf{x}_{i,j}(\mathbf{M'})) \text{,}
			\end{equation}
			where $1\leq i <a$.
			
			Attributed to $\mathbf{s}_i \neq \mathbf{s}'_i$ for $a \leq i\leq n/k$,  $\operatorname{f}(\mathbf{x}_{i,j}(\mathbf{M}))$ and $\operatorname{f}(\mathbf{x}_{i,j}(\mathbf{M'}))$ follow the uniform distribution over $[0,2^c-1]$, i.d.d (independent and identically distribution). 
			
			
			\section{Proof of lemma \ref{lemma1}}\label{rotationmatrix}
			To simplify the probability $\mathrm{Pr} \left( \mathbf{v}^{L_a} {\left(\mathbf{v}^{L_a} + 2{\mathbf{N}}^{L_a}\right)}^{\mathrm{T}} \leq 0 \right)$ in (\ref{integral19}), we introduce a rotation matrix for the $L_a$-dimensions hyperspace.
			The rotation matrix is denoted as
			\begin{equation} \label{eq17}
				\mathbf{A} = 
				\begin{bmatrix}
					\mathbf{A}_1 \\
					\vdots \\
					\mathbf{A}_{L_a}
				\end{bmatrix}
				\in \mathbb{R}^{L_a \times L_a} \text{,}
			\end{equation}
			which satisfies
			\begin{equation} \label{eq18}
				\mathbf{A}{\left[\mathbf{v}^{L_a}\right]}^{\mathrm{T}} = {\bigg[ \left\|\mathbf{v}^{L_a}\right\|,\underbrace{0,\cdots,0}_{L_a-1} \bigg]}^{\mathrm{T}} \text{.}
			\end{equation}
			
			As a rotation matrix, $\mathbf{A}$ owns the property that ${\mathbf{A}}^{\mathrm{T}} \mathbf{A} = \mathbf{I}_{L_a}$, which can be used to rewrite the probability as
			\begin{equation} \label{eq19}
				\begin{split}
					& \mathrm{Pr} \left( \mathbf{v}^{L_a} {\left(\mathbf{v}^{L_a} + 2{\mathbf{N}}^{L_a}\right)}^{\mathrm{T}} \leq 0 \right) \\
					&= \mathrm{Pr} \left( \mathbf{v}^{L_a} \mathbf{I}_{L_a} {\left(\mathbf{v}^{L_a} + 2{\mathbf{N}}^{L_a}\right)}^{\mathrm{T}} \leq 0 \right) \\
					&= \mathrm{Pr} \left( \mathbf{v}^{L_a} {\mathbf{A}}^{\mathrm{T}} \mathbf{A} {\left(\mathbf{v}^{L_a} + 2{\mathbf{N}}^{L_a}\right)}^{\mathrm{T}} \leq 0 \right) \\
					&= \mathrm{Pr} \left( {\left[ \mathbf{A}{\left[\mathbf{v}^{L_a}\right]}^{\mathrm{T}} \right]}^{\mathrm{T}} \left( \mathbf{A}{\left[\mathbf{v}^{L_a}\right]}^{\mathrm{T}} + 2\mathbf{A}{\left[\mathbf{N}^{L_a}\right]}^{\mathrm{T}} \right) \leq 0 \right) \text{.}
				\end{split}
			\end{equation}
			
			Applying (\ref{eq17}) and (\ref{eq18}) in (\ref{eq19}) results in
			\begin{equation} \label{eq20}
				\begin{split}
					&\mathrm{Pr} \left( {\left\| \mathbf{v}^{L_a} \right\|}^2 + 2\left\| \mathbf{v}^{L_a} \right\| \cdot \mathbf{A}_1 {\left[\mathbf{N}^{L_a}\right]}^{\mathrm{T}} \leq 0 \right) \\
					&= \mathrm{Pr} \left( \mathbf{A}_1 {\left[\mathbf{N}^{L_a}\right]}^{\mathrm{T}} \leq -\frac{\left\| \mathbf{v}^{L_a} \right\|}{2} \right) \text{,}
				\end{split}
			\end{equation}
			the RHS of which is given by the following equivalent form:
			\begin{equation} \label{eq21}
				\begin{split}
					&\mathrm{Pr} \left( \mathbf{A}_1 {\left[\mathbf{N}^{L_a}\right]}^{\mathrm{T}} \leq -\frac{\left\| \mathbf{v}^{L_a} \right\|}{2} \right) \\
					&= \overbrace{\int \cdots \int}^{L_a}_{_{\mathbf{A}_1 {\left[\mathbf{N}^{L_a}\right]}^{\mathrm{T}} \leq -\frac{\left\| \mathbf{v}^{L_a} \right\|}{2}}} \frac{1}{{\left(\sqrt{2\pi}\sigma\right)}^{L_a}} \mathrm{e}^{-\frac{{\left\| \mathbf{N}^{L_a} \right\|}^2}{2\sigma^2}} \prod_{i=a}^{n/k} \prod_{j=1}^{L} \mathrm{d}n_{i,j} \text{.}
				\end{split}
			\end{equation}
			
			In order to further simplify the $L_a$-fold integral above, we introduce another rotation matrix $\mathbf{B}$, which is denoted as
			\begin{equation} \label{eq22}
				\mathbf{B} = 
				\begin{bmatrix}
					\mathbf{B}_1 \\
					\vdots \\
					\mathbf{B}_{L_a}
				\end{bmatrix}
				\in \mathbb{R}^{L_a \times L_a} \text{,}
			\end{equation}
			such that
			\begin{equation} \label{eq23}
				\mathbf{B}{\left[\mathbf{N}^{L_a}\right]}^{\mathrm{T}} = {\bigg[ \mathbf{A}_1{\left[\mathbf{N}^{L_a}\right]}^{\mathrm{T}},\underbrace{0,\cdots,0}_{L_a-1} \bigg]}^{\mathrm{T}} \text{,}
			\end{equation}
			and that
			\begin{equation} \label{eq24}
				{\left[\mathbf{N}^{L_a}\right]}^{\mathrm{T}} = \mathbf{B}^{\mathrm{T}}{\left[\mathbf{n}^{L_a}\right]}^{\mathrm{T}} \text{,}
			\end{equation}
			where $\mathbf{n}^{L_a}$ is the transformed vector composed of $n'_{i,j}$ with $a \leq i \leq n/k, 1 \leq j \leq L$.
			
			From (\ref{eq23}) we can get that
			\begin{equation} \label{eq25}
				\mathbf{B}_1{\left[\mathbf{N}^{L_a}\right]}^{\mathrm{T}} = \mathbf{A}_1 {\left[\mathbf{N}^{L_a}\right]}^{\mathrm{T}} \text{.}
			\end{equation}
			
			Substitute (\ref{eq24}) into (\ref{eq25}), we have that
			\begin{align} \label{eq30}
					&\mathbf{A}_1 {\left[\mathbf{N}^{L_a}\right]}^{\mathrm{T}} = \mathbf{B}_1{\left[\mathbf{N}^{L_a}\right]}^{\mathrm{T}} \notag\\
					&= 
					\mathbf{B}_1 \mathbf{B}^{\mathrm{T}} {\left[\mathbf{n}^{L_a}\right]}^{\mathrm{T}} \notag\\
					&= 
					\begin{bmatrix}
						\mathbf{B}_1 \mathbf{B}_1^{\mathrm{T}} & \mathbf{B}_1 \mathbf{B}_2^{\mathrm{T}} & \cdots & \mathbf{B}_1 \mathbf{B}_{L_a}^{\mathrm{T}} 
					\end{bmatrix} 
					{\left[\mathbf{n}^{L_a}\right]}^{\mathrm{T}} \notag\\
					& \overset{(c)}{=} \left[ 1, \underbrace{0,\cdots,0}_{L_a-1} \right] {\left[\mathbf{n}^{L_a}\right]}^{\mathrm{T}} = n'_{a,1} \text{,}
			\end{align}
			where (c) is obtained from the property that rotation matrix fulfills $\mathbf{B}{\mathbf{B}}^{\mathrm{T}} = \mathbf{I}_{L_a}$. 
			
			From (\ref{eq24}) we can also derive that
			\begin{equation} \label{eqn}
				\begin{split}
					{\left\| \mathbf{N}^{L_a} \right\|}^2 &= {\left[\mathbf{N}^{L_a}\right]}^{\mathrm{T}}\mathbf{N}^{L_a}\\
					&= \mathbf{n}^{L_a}\mathbf{B}\mathbf{B}^{\mathrm{T}}{\left[\mathbf{n}^{L_a}\right]}^{\mathrm{T}} \overset{(c)}{=} {\left\| \mathbf{n}^{L_a} \right\|}^2 \text{,}
				\end{split}
			\end{equation}
			where (c) establishes by adopting $\mathbf{B} {\mathbf{B}}^{\mathrm{T}} = \mathbf{I}_{L_a}$.		
			
			Together with (\ref{eq30}) and (\ref{eqn}), (\ref{eq21}) can be rewritten as \cite{Jacobian}
			\begin{equation} \label{eq26}
				\small
				\begin{split}
					&\overbrace{\int \cdots \int}^{L_a}_{\mathbf{A}_1 {\left[\mathbf{N}^{L_a}\right]}^{\mathrm{T}} \leq -\frac{\left\| \mathbf{v}^{L_a} \right\|}{2}} \frac{1}{{\left(\sqrt{2\pi}\sigma\right)}^{L_a}} \mathrm{e}^{-\frac{{\left\| \mathbf{N}^{L_a} \right\|}^2}{2\sigma^2}} \prod_{i=a}^{n/k} \prod_{j=1}^{L} \mathrm{d}n_{i,j} \\
					&= \overbrace{\int \cdots \int}^{L_a}_{n'_{a,1} \leq -\frac{\left\| \mathbf{v}^{L_a} \right\|}{2}} \frac{1}{{\left(\sqrt{2\pi}\sigma\right)}^{L_a}} \mathrm{e}^{-\frac{ {\left\| \mathbf{n}^{L_a} \right\|}^2 }{2\sigma^2}} \left| \mathbf{J}\left(\mathbf{n}^{L_a}\right) \right|\prod_{i=a}^{n/k} \prod_{j=1}^{L} \mathrm{d}n'_{i,j} \text{,}
				\end{split}
			\end{equation}
			in which $\mathbf{J}\left(\mathbf{n}^{L_a}\right)$ is the Jacobi matrix, behaving as
			\begin{equation} \label{eq27}
				\mathbf{J}\left(\mathbf{n}^{L_a}\right) = 
				\begin{bmatrix}
					\frac{\partial n'_{a,1}}{n_{a,1}} & \cdots & \frac{\partial n'_{a,1}}{n_{n/k,L}} \\
					\vdots & \ddots & \vdots \\
					\frac{\partial n'_{n/k,L}}{n_{a,1}} & \cdots & \frac{\partial n'_{n/k,L}}{n_{n/k,L}}
				\end{bmatrix} \text{.}
			\end{equation}
			
			Since ${\left[\mathbf{N}^{L_a}\right]}^{\mathrm{T}} = \mathbf{B}^{\mathrm{T}}{\left[\mathbf{n}^{L_a}\right]}^{\mathrm{T}}$, we get ${\left[\mathbf{n}^{L_a}\right]}^{\mathrm{T}} = \mathbf{B}{\left[\mathbf{N}^{L_a}\right]}^{\mathrm{T}}$, therefore the Jacobi matrix behaves as
			\begin{equation} \label{eq28}
				\mathbf{J}\left(\mathbf{n}^{L_a}\right) = \mathbf{B} \text{.}
			\end{equation}
			
			Recalling that $\mathbf{B} {\mathbf{B}}^{\mathrm{T}} = \mathbf{I}_{L_a}$, we have 
			\begin{equation} \label{eq29}
				\begin{split}
					\left| \mathbf{J}\left(\mathbf{n}^{L_a}\right) \right| &= \sqrt{{\left| \mathbf{J}\left(\mathbf{n}^{L_a}\right) \right|}^2} \\
					&= \sqrt{{\left|\mathbf{B}\right|}^2} = \sqrt{\left| \mathbf{B}{\mathbf{B}}^{\mathrm{T}} \right|} = 1 \text{.}
				\end{split}
			\end{equation}
			
			Then, by substituting (\ref{eq29}) into (\ref{eq26}), we could verify that
			\begin{equation} \label{eq31}
				\small
				\begin{split}
					&\overbrace{\int \cdots \int}^{L_a}_{n'_{a,1} \leq -\frac{\left\| \mathbf{v}^{L_a} \right\|}{2}} \frac{1}{{\left(\sqrt{2\pi}\sigma\right)}^{L_a}} \mathrm{e}^{-\frac{ {\left\| \mathbf{n}^{L_a} \right\|}^2 }{2\sigma^2}} \left| \mathbf{J}\left(\mathbf{n}^{L_a}\right) \right|\prod_{i=a}^{n/k} \prod_{j=1}^{L} \mathrm{d}n'_{i,j} \\
					&= \overbrace{\int \cdots \int}^{L_a}_{n'_{a,1} \leq -\frac{\left\| \mathbf{v}^{L_a} \right\|}{2}} \frac{1}{{\left(\sqrt{2\pi}\sigma\right)}^{L_a}} \mathrm{e}^{-\frac{{\left\| \mathbf{n}^{L_a} \right\|}^2}{2\sigma^2}} \prod_{i=a}^{n/k} \prod_{j=1}^{L} \mathrm{d}n'_{i,j} \\
					&= \int_{-\infty}^{-\frac{\left\| \mathbf{v}^{L_a} \right\|}{2}} \frac{1}{{\sqrt{2\pi}\sigma}} \mathrm{e}^{-\frac{{\left( n'_{a,1} \right)}^2}{2\sigma^2}} \mathrm{d}n'_{a,1} = Q\left(\frac{ \left\| \mathbf{v}^{L_a} \right\| }{2\sigma}\right) \text{.}
				\end{split}
			\end{equation}
			
			\section{Proof of Lemma \ref{lemma2}} \label{proofray}
			The central problem is to solve out the integral of $h$ in the RHS of (\ref{eqreplace}). 
			Recall that $h$ follows Rayleigh distribution, whose PDF is given as	
			\begin{equation} \label{eq6}
				g_{_1}\left(h\right) = \frac{2h}{\Omega} \mathrm{exp} \left( \frac{-h^2}{\Omega} \right) \text{.}
			\end{equation}	
		
			Insert (\ref{eq6}) into the RHS of (\ref{eqreplace}), the integral is transformed as
			\begin{equation}
				\begin{split}
					& \int_{h} {\mathrm{exp}}{\left(-\frac{{h^2}{u^2}}{{8}{\sigma^2}{\sin^2}\theta}\right)} g_{_1}(h)\mathrm{d}h \\
					&= \int_{0}^{+\infty} {\mathrm{exp}}{\left(-\frac{{h^2}{u^2}}{{8}{\sigma^2}{\sin^2}\theta}\right)} \cdot \frac{2h}{\Omega} \mathrm{exp} \left( \frac{-h^2}{\Omega} \right)\mathrm{d}h \\
					& \xlongequal{z = \frac{u^2}{8{{\sigma}^{2}}{{\sin}^{2}}\theta}} \frac{1}{\Omega} \int_{0}^{+\infty} \mathrm{exp}\left[- \left( z+\frac{1}{\Omega} \right) {h^2}\right] \mathrm{d}{h^2} \\
					&=\frac{1}{z\Omega+1} =\frac{8{\sigma^2}{\sin^2}\theta}{{\Omega}u^2+8{\sigma^2}{\sin^2}\theta} \text{,}
				\end{split}
			\end{equation}
			with which the RHS of (\ref{eqreplace}) becomes
			\begin{equation} \label{eq64}
				\begin{split}
					& \sum_{u} p_{_U}(u) \int_{h} {\mathrm{exp}}{\left(-\frac{{h^2}{u^2}}{{8}{\sigma^2}{\sin^2}\theta}\right)} g_1(h)\mathrm{d}h\\
					&= \sum_{u} p_{_U}(u) \cdot \frac{8{\sigma^2}{\sin^2}\theta}{{\Omega}u^2+8{\sigma^2}{\sin^2}\theta} \text{.}
				\end{split}
			\end{equation} 
			
			Recall that $U = U_{a,1} = \operatorname{f}(\mathbf{x}_{a,1}(\mathbf{M})) - \operatorname{f}(\mathbf{x}_{a,1}(\mathbf{M'}))$. Appendix \ref{property} draws the conclusion that both $\operatorname{f}(\mathbf{x}_{a,1}(\mathbf{M}))$ and $\operatorname{f}(\mathbf{x}_{a,1}(\mathbf{M'}))$ follow the uniform distribution over $[0,2^c-1]$, i.d.d. As a result, (\ref{eq64}) extends to
			\begin{equation}
				\begin{split}
					&\sum_{u} p_{_U}(u) \cdot \frac{8{\sigma^2}{\sin^2}\theta}{{\Omega}u^2+8{\sigma^2}{\sin^2}\theta} \\
					&\xlongequal[Y = \operatorname{f}(\mathbf{x}_{a,1}(\mathbf{M'}))]{X = \operatorname{f}(\mathbf{x}_{a,1}(\mathbf{M}))} \sum_{i \in \Psi} \sum_{j \in \Psi} p_{_X}(i)p_{_Y}(j) \frac{8{\sigma}^2{\sin}^2{\theta}}{{\Omega}(i-j)^2+8{\sigma}^2{\sin}^2{\theta}} \\
					&= \sum_{i \in \Psi} \sum_{j \in \Psi} 2^{-2c} \frac{8{\sigma}^2{\sin}^2{\theta}}{{\Omega}(i-j)^2+8{\sigma}^2{\sin}^2{\theta}} \text{.}
				\end{split}
			\end{equation}
			
			\section{Proof Of Monotony} \label{monotony}
			\begin{equation} \label{eq66}
				\small
				\mathcal{F}_{\text{Rayleigh}}(\theta\text{;} \sigma\text{,}L_a) =
				{\left( \sum_{i \in \Psi} \sum_{j \in \Psi} 2^{-2c} \frac{8{\sigma}^2{\sin}^2{\theta}}{{\Omega}(i-j)^2+8{\sigma}^2{\sin}^2{\theta}} \right)}^{L_a} \text{.}
			\end{equation}
			
			The monotony of $\mathcal{F}_{\text{Rayleigh}}(\theta\text{;} \sigma\text{,}L_a)$ equals to the monotony of
			\begin{equation}
				F_{\text{Rayleigh}}(\theta) \triangleq \frac{8{\sigma}^2{\sin}^2{\theta}}{{\Omega}(i-j)^2+8{\sigma}^2{\sin}^2{\theta}}\text{,}
			\end{equation}
			where $\theta \in [0,\frac{\pi}{2}]$.
			
			Denoting $8{\sigma}^2 = b$ and ${\Omega}(i-j)^2 = k$, the derivative of $F(\theta)$ is given by
			\begin{equation} \label{derivation}
				\frac{\mathrm{d} F_{\text{Rayleigh}}(\theta)}{\mathrm{d} \theta} = 2kb \frac{\mathrm{sin}\theta \mathrm{cos}\theta}{{(k+\mathrm{sin}^2\theta)}^2} \geq 0 \text{,}
			\end{equation}
			which indicates the monotonically increasing property of function $F_{\text{Rayleigh}}(\theta)$.
			
			As a consequence, $\mathcal{F}_{\text{Rayleigh}}(\theta\text{;} \sigma\text{,}L_a)$ is a monotonically increasing function.
			
			\begin{equation} \label{eqq50}
				\begin{split}
					& \mathcal{F}_{\text{Nakagami}}(\theta\text{;} \sigma\text{,}L_a) \\
					&= {\left( \sum_{i \in \Psi} \sum_{j \in \Psi} 2^{-2c}
						{\left(\frac{8m{\sigma}^2{\sin}^2{\theta}}{{\Omega}(i-j)^2+8m{\sigma}^2{\sin}^2{\theta}}\right)}^m \right)}^{L_a} \text{.}
				\end{split}
			\end{equation}	
			
			Because $m \geq 0.5$, the monotony of $\mathcal{F}_{\text{Nakagami}}(\theta\text{;} \sigma\text{,}L_a)$ equals to the monotony of 
			\begin{equation}
				F_{\text{Nakagami}}(\theta) \triangleq \frac{8m{\sigma}^2{\sin}^2{\theta}}{{\Omega}(i-j)^2+8m{\sigma}^2{\sin}^2{\theta}} \textbf{,}
			\end{equation}
			where $\theta \in [0,\frac{\pi}{2}]$.
			
			Similar to the proof for $\mathcal{F}_{\text{Rayleigh}}(\theta\text{;} \sigma\text{,}L_a)$, we denote that $8m{\sigma}^2 = b$ and ${\Omega}(i-j)^2 = k$. Since the result is the same as (\ref{derivation}), $\mathcal{F}_{\text{Nakagami}}(\theta\text{;} \sigma\text{,}L_a)$ is also a monotonically increasing function.
						
			\begin{equation} \label{eqq62}
				\small
				\begin{split}
					& \mathcal{F}_{\text{Rician}}(\theta\text{;} \sigma\text{,}L_a) \\
					&= \Bigg( \sum_{i \in \Psi} \sum_{j \in \Psi} 2^{-2c} \frac{8(K+1){\sigma}^2{\sin}^2{\theta}}{{\Omega}(i-j)^2+8(K+1){\sigma}^2{\sin}^2{\theta}}  \\
					& \qquad \cdot {\mathrm{exp} \left( \frac{8K(K+1){\sigma}^2{\sin}^2{\theta}}{{\Omega}(i-j)^2+8(K+1){\sigma}^2{\sin}^2{\theta}} -K\right) \Bigg)}^{L_a} \text{,}
				\end{split}
			\end{equation}
			
			Because exponential function is monotonically increasing, the monotony of $\mathcal{F}_{\text{Rician}}(\theta\text{;} \sigma\text{,}L_a)$ equals to the monotony of
			\begin{equation}
				F_{\text{Rician}}(\theta) \triangleq \frac{8(K+1){\sigma}^2{\sin}^2{\theta}}{{\Omega}(i-j)^2+8(K+1){\sigma}^2{\sin}^2{\theta}}\text{.}
			\end{equation}
			where $\theta \in [0,\frac{\pi}{2}]$.
			
			Similar to the proof for $\mathcal{F}_{\text{Rayleigh}}(\theta\text{;} \sigma\text{,}L_a)$, we denote $8(K+1){\sigma}^2 = b$ and ${\Omega}(i-j)^2 = k$. As the result is the same as $(\ref{derivation}), \mathcal{F}_{\text{Rician}}(\theta\text{;} \sigma\text{,}L_a)$ is also a monotonically increasing function.
			
			\section{Proof of Theorem \ref{theorem2}} \label{proofnaka}
			
			For Spinal codes over the Nakagami-m fading channel, the only difference is to replace Rayleigh distribution with Nakagami-m distribution. The PDF Nakagami-m distribution is given by
			\begin{equation} \label{eq46}
				g_{_2}(h) = \frac{2m^m}{\Gamma(m)\Omega^m} h^{2m-1} \text{exp} \left( \frac{-mh^2}{\Omega} \right) \text{,}
			\end{equation}
			where $h$ is the channel fading parameter, $m \geq 0.5$ and the Gamma function is defined as	
			\begin{equation} \label{eq47}
				\Gamma(m) = {\int_{0}^{+\infty}} \mathrm{e}^{-t} t^{m-1} \mathrm{d}t \text{.}
			\end{equation}
			
			Substitute (\ref{eq46}) into the integral of (\ref{eqreplace}) and denote $z = \frac{u^2}{8{{\sigma}^{2}}{{\sin}^{2}}\theta}$, after which the integral of $h$ is transformed as
			\begin{equation}
				\begin{split} \label{eq48}
					& \int_{h} {\mathrm{exp}}{\left(-\frac{{h^2}{u^2}}{{8}{\sigma^2}{\sin^2}\theta}\right)} g_{_2}(h)\mathrm{d}h \\
					&= \int_{0}^{+\infty} \mathrm{e}^{-{z}{h^2}} \cdot \frac{2m^m}{\Gamma(m)\Omega^m} h^{2m-1} \mathrm{e}^ {-\frac{mh^2}{\Omega}} \mathrm{d}h\\
					& \xlongequal[h = \sqrt{\frac{t}{z+\frac{m}{\Omega}}}]{t = \left( z+\frac{m}{\Omega} \right)h^2} \frac{m^m}{\Gamma(m){(z\Omega+m)}^m} {\int_{0}^{+\infty}} \mathrm{e}^{-t} t^{m-1} \mathrm{d}t \\
					&= \frac{m^m}{{(z\Omega+m)}^m}  \text{.}
				\end{split}
			\end{equation}
			
			As mentioned above, $z = \frac{u^2}{8{{\sigma}^{2}}{{\sin}^{2}}\theta}$, substitute $z$ into (\ref{eq48}):
			\begin{equation} \label{eq49}
				\begin{split}
					\frac{m^m}{{(z\Omega+m)}^m} = {\left(\frac{8m{\sigma}^2{\sin}^2{\theta}}{{\Omega}u^2+8m{\sigma}^2{\sin}^2{\theta}}\right)}^m \text{.}
				\end{split}
			\end{equation}
			
			Then the process is similar to Theorem \ref{theorem1}, by denoting that
			\begin{equation} \label{eq50}
				\begin{split}
					& \mathcal{F}_{\text{Nakagami}}(\theta\text{;} \sigma\text{,}L_a) \\
					&= {\left( \sum_{i \in \Psi} \sum_{j \in \Psi} 2^{-2c}
						{\left(\frac{8m{\sigma}^2{\sin}^2{\theta}}{{\Omega}(i-j)^2+8m{\sigma}^2{\sin}^2{\theta}}\right)}^m \right)}^{L_a} \text{,}
				\end{split}
			\end{equation}
			which is monotonically increasing proved in Appendix \ref{monotony}, and	that
			\begin{equation} \label{eq51}
				\mathscr{F}_{\text{Nakagami}} \left(L_a \text{,} \sigma \right) = \sum_{r=1}^{N} b_r\mathcal{F}_{\text{Nakagami}}(\theta_r \text{;} \sigma\text{,}L_a)\text{,}
			\end{equation}
			we can get the explicit bound.

			\section{Proof of Theorem \ref{theorem3}} \label{proofrice}
			
			For Spinal codes over the Rician fading channel, the only difference is to replace Rayleigh distribution with Rician distribution. The PDF of Rician distribution is given by
			\begin{equation} \label{eq56}
				\small
				g_{_3}(h) = \frac{2(K+1)h}{{\Omega}\, {\mathrm{exp} \left( K+\frac{(K+1){h^2}}{\Omega} \right)}} I_0 \left( 2\sqrt{\frac{K(K+1)}{\Omega}} h\right) \text{,}
			\end{equation}
			where $h$ is the channel fading parameter, and $I_0(\cdot)$ is the modified Bessel function of the first kind with order zero, given by
			\begin{equation} \label{eq57}
				I_0(x) = \sum_{k=0}^{+\infty} \frac{{(x^2/4)}^k}{k!\Gamma(k+1)} \text{.}
			\end{equation}

			Substitute (\ref{eq56}) into (\ref{eqreplace}) and denote $z = \frac{u^2}{8{{\sigma}^{2}}{{\sin}^{2}}\theta}$, then the integral of $h$ is transformed as
			\begin{equation}
				\small
				\begin{split} \label{eq58}
					& \int_{h} {\mathrm{exp}}{\left(-\frac{{h^2}{u^2}}{{8}{\sigma^2}{\sin^2}\theta}\right)} g_{_3}(h)\mathrm{d}h = \int_{0}^{+\infty} \mathrm{e}^{-{z}{h^2}} g_{_3}(h) \mathrm{d}h\\
					&= \frac{2(K+1)}{{\Omega} \mathrm{e}^{K}} {\int_{0}^{+\infty}} {h\mathrm{e}^{- \left( \frac{K+1}{\Omega}+z \right)h^2} I_0 \left( 2\sqrt{\frac{K(K+1)}{\Omega}}h \right) \mathrm{d}h} \\
					&= \frac{2}{\mathrm{e}^{K}} \sum_{m=0}^{+\infty} \frac{ K^m{(K+1)}^{m+1} }{m!\Gamma(m+1)\Omega^{m+1}} {\int_{0}^{+\infty}} h^{2m+1} \mathrm{e}^{- \left( \frac{K+1}{\Omega}+z \right)h^2} \mathrm{d}h \text{.}
				\end{split}
			\end{equation}
			
			Let $t = \left( z+\frac{K+1}{\Omega} \right)h^2 $, $ h = \sqrt{\frac{t}{z+\frac{K+1}{\Omega}}}$, we have
			\begin{equation}
				\begin{split} \label{eq59}
					& {\int_{0}^{+\infty}} h^{2m+1} \mathrm{e}^{- \left( \frac{K+1}{\Omega}+z \right)h^2} \mathrm{d}h \\
					&= \int_{0}^{+\infty} \mathrm{e}^{-t} \cdot \frac{t^m}{2{\left( z + \frac{K+1}{\Omega} \right)}^{m+1}} \mathrm{d}t \\
					&= \frac{\Gamma(m+1)}{2{\left( z + \frac{K+1}{\Omega} \right)}^{m+1}} \text{.}
				\end{split}
			\end{equation}
			
			With (\ref{eq59}) in hand, (\ref{eq58}) can be rewritten as
			\begin{equation} \label{eq60}
				\begin{split}
					& \int_{h} {\mathrm{exp}}{\left(-\frac{{h^2}{u^2}}{{8}{\sigma^2}{\sin^2}\theta}\right)} g_{_3}(h)\mathrm{d}h \\
					&= \frac{1}{K\mathrm{e}^{K}} \cdot \frac{K(K+1)}{{\Omega}z+K+1} \cdot \sum_{m=0}^{+\infty} \frac{1}{m!}{\left( \frac{K(K+1)}{{\Omega}z+K+1} \right)}^{m}  \\
					&= \frac{1}{K\mathrm{e}^{K}} \cdot \frac{K(K+1)}{{\Omega}z+K+1} \cdot \mathrm{exp}\left( \frac{K(K+1)}{{\Omega}z+K+1}  \right) \text{.}
				\end{split}
			\end{equation}
			
			Now that $z = \frac{u^2}{8{{\sigma}^{2}}{{\sin}^{2}}\theta}$, (\ref{eq60}) comes into
			\begin{equation}
				\small
				\begin{split}
					& \frac{1}{K\mathrm{e}^{K}} \cdot \frac{K(K+1)}{{\Omega}z+K+1} \cdot \mathrm{exp}\left( \frac{K(K+1)}{{\Omega}z+K+1}  \right) \\
					&= \frac{8\mathrm{e}^{-K}(K+1){\sigma^2}{\sin^2}\theta}{{\Omega}u^2+8(K+1){\sigma^2}{\sin^2}\theta} \mathrm{exp} \left( \frac{8K(K+1){\sigma^2}{\sin^2}\theta}{{\Omega}u^2+8(K+1){\sigma^2}{\sin^2}\theta} \right) \text{.}
				\end{split}
			\end{equation}
			
			Then the process is similar to Theorem \ref{theorem1}. By denoting that
			\begin{equation} \label{eq62}
				\small
				\begin{split}
					& \mathcal{F}_{\text{Rician}}(\theta\text{;} \sigma\text{,}L_a) \\
					&= \Bigg( \sum_{i \in \Psi} \sum_{j \in \Psi} 2^{-2c} \frac{8(K+1){\sigma}^2{\sin}^2{\theta}}{{\Omega}(i-j)^2+8(K+1){\sigma}^2{\sin}^2{\theta}}  \\
					& \qquad \cdot {\mathrm{exp} \left( \frac{8K(K+1){\sigma}^2{\sin}^2{\theta}}{{\Omega}(i-j)^2+8(K+1){\sigma}^2{\sin}^2{\theta}} -K\right) \Bigg)}^{L_a} \text{,}
				\end{split}
			\end{equation}
			which is monotonically increasing proved in Appendix \ref{monotony}, and	that	
			\begin{equation} \label{eq63}
				\mathscr{F}_{\text{Rician}} \left(L_a \text{,} \sigma \right) = \sum_{r=1}^{N} b_r\mathcal{F}_{\text{Rician}}(\theta_r \text{;} \sigma\text{,}L_a) \text{,}
			\end{equation}
			we can get the explicit bound.
		\end{appendices}
		
	\end{document}